\documentclass[letterpaper]{article}
\usepackage{aaai2026} %
\usepackage{times} %
\usepackage{helvet} %
\usepackage{courier} %
\usepackage[hyphens]{url} %
\usepackage{graphicx} %
\usepackage{graphicx} %
\usepackage{natbib} %
\usepackage{caption} %
\usepackage{algorithm}
\nocopyright
\usepackage{newfloat}
\usepackage{listings}
\DeclareCaptionStyle{ruled}{labelfont=normalfont,labelsep=colon,strut=off} %
\floatstyle{ruled}
\newfloat{listing}{tb}{lst}{}
\floatname{listing}{Listing}
\usepackage[dvipsnames]{xcolor}
\usepackage{tabularx}
\usepackage{amsmath}
\usepackage{amsthm}
\usepackage{amssymb,amsfonts,textcomp}
\usepackage{pifont}
\usepackage{latexsym}
\usepackage{graphicx} 
\usepackage{multirow}
\usepackage{microtype}
\usepackage{rotating}
\usepackage{t1enc}
\usepackage{url}
\usepackage{comment}
\usepackage{paralist}
\usepackage{booktabs}
\usepackage{bm}
\usepackage[]{units}
\usepackage{tikz}
\usepackage{natbib}
\usepackage{mathtools}

\usepackage{eso-pic}

\usepackage[nameinlink]{cleveref}
\usepackage{algpseudocode}

\usepackage{thmtools, thm-restate}
\usepackage[autostyle, english = american]{csquotes}
\MakeOuterQuote{"}

\newtheorem{definition}{Definition}
\newtheorem{proposition}{Proposition}

\newtheorem{theorem}{Theorem}

\Crefname{claim}{claim}{claims}

\theoremstyle{definition}
\newtheorem{example}{Example}

\allowdisplaybreaks

\DeclareMathOperator*{\argmin}{arg\,min}

\RequirePackage{ifthen,calc}
\newsavebox{\ffbox}

\usepackage{pifont} %
\newcommand{\cmark}{\textcolor{green!30!gray}{\ding{51}}} %

\newcommand{\naturals}{{{\mathbb{N}}}}

\newcommand{\val}{\zeta}

\theoremstyle{definition}

\newcommand{\alg}[1]{\textproc{#1}}
\newcommand{\algcomp}[1]{$\overline{\mbox{\alg{#1}}}$}
\newcommand{\GM}{\algcomp{MES}}

\newcommand{\algmu}[1]{\textproc{#1$_\mu$}}
\newcommand{\algmuprime}[1]{\textproc{#1$_{\mu'}$}}
\newcommand{\algcompmu}[1]{$\overline{\mbox{\algmu{#1}}}$}
\newcommand{\GMmu}{\algcompmu{MES}}

\newcommand{\Ratio}{W}
\newcommand{\Ratiohat}{\hat{W}}

\newif\ifcomments %
\commentsfalse

\usepackage{enumitem} 

\title{Utilitarian Guarantees for the Method of Equal Shares}

\author {
    Anton Baychkov\textsuperscript{\rm 1},
    Markus Brill\textsuperscript{\rm 1},
    Jannik Peters\textsuperscript{\rm 2}
}
\affiliations {
    \textsuperscript{\rm 1}University of Warwick, Coventry, UK \\ 
    \textsuperscript{\rm 2}National University of Singapore, Singapore
    \\
    anton.baychkov@warwick.ac.uk, markus.brill@warwick.ac.uk, peters@nus.edu.sg
}

\begin{document}

\maketitle

\begin{abstract}

In recent years, research in Participatory Budgeting (PB) has put a greater emphasis on rules satisfying notions of fairness and proportionality, with the \emph{Method of Equal Shares} (MES) being a prominent example. However, proportionality can come at a cost to the total utilitarian welfare. Our work formalizes this relationship, by deriving minimum utilitarian welfare guarantees for MES for a subclass of satisfaction functions called DNS functions, which includes two of the most popular ways of measuring a voter’s utility in the PB setting: considering (1) the total cost of approved projects or (2) the total number of those projects. Our results are parameterized in terms of minimum and maximum project costs, which allows us to improve on the mostly negative results found in prior studies, and reduce to the existing multiwinner guarantee when project costs are equal. We show that our guarantees are asymptotically tight for rules satisfying \emph{Extended Justified Representation up to one project}, showing that no proportional rule can achieve a better utilitarian guarantee than MES. %

\end{abstract}

\section{Introduction}

Participatory budgeting (PB) is a key topic in computational social choice \citep{ReMa23a}. PB enables voters to directly vote on how public money is spent. This is often accomplished by allowing the voters to express their preferences over a set of available projects. These preferences are then aggregated using a voting rule selecting a subset of the projects that fits within a pre-defined budget limit \citep{WMT21a}. 

In practice, the voting process is usually conducted via \emph{approval ballots}, i.e., each voter can indicate their support for the subset of projects they like, without the need to specify a ranking among them. These ballots are usually aggregated using a simple \alg{Greedy} method: projects are processed in order of decreasing welfare-to-cost ratio (i.e., ``bang for buck''), where a project's welfare is determined by the total satisfaction it provides to voters. In the most common formulation, where satisfaction is measured by project cost, this reduces to ordering projects by decreasing vote count. When processing projects in this order, each project that fits into the remaining budget is added to the outcome. This method is simple and easy to explain, and greedily maximizes the \textit{utilitarian welfare} %
of the outcome. %

However, \alg{Greedy} has a major downside: it is not \emph{proportional}. To see this, consider a scenario where 60\% of voters approve only infrastructure projects while 40\% approve only leisure projects. \alg{Greedy} would spend all the budget on infrastructure projects, leaving the minority completely unrepresented despite their legitimate claim to a proportional share of the budget.
Meanwhile, the recently introduced \textit{Method of Equal Shares} (\alg{MES}) \citep{PPS21a} yields a fairer outcome by splitting the total budget equally among voters, and allowing groups of voters to ``buy'' projects using their allocated funds.
While this leads to proportional outcomes\,---\,as formalized by variants of the \textit{extended justified representation} (EJR) axiom\,---\,\alg{MES} can produce outcomes with low utilitarian welfare. %

This  tradeoff between utilitarian and proportionality objectives leads to the question of how much utilitarian welfare can still be guaranteed by proportional rules in general, and by MES in particular. To answer this question, we study \emph{utilitarian guarantees}, i.e., lower bounds on the ratio between the utilitarian welfare provided by a rule on a given instance and the maximal achievable utilitarian welfare for that instance.\footnote{The inverse of this ratio (for ``fair'' voting rules) is sometimes referred as the \emph{price of fairness} (e.g., \citealp{EFI+24a}); a lower ratio corresponds to a higher price.}

For PB, the definition of utilitarian guarantees immediately runs into two fundamental issues. 
First, it is unclear how voter satisfaction (or welfare) should be defined in this setting, with two differing definitions having emerged in the literature: 
(i) \emph{cost satisfaction} assumes that a voter's satisfaction is measured in terms of the amount of money spent on projects they approve; whereas 
(ii) \emph{cardinality satisfaction} assumes that the satisfaction of a voter is given by the number of approved projects in the outcome. 
Second, the meaningfulness of utilitarian guarantees depends critically on how they are parameterized. The PB setting offers a variety of potential parameters: budget size, project cost distributions, voter preferences, voter turnout, etc. However, the value of guarantees that depend on uncontrollable factors (such as voter turnout or preference diversity) is questionable, as these factors may grow large or vary unpredictably. Instead, meaningful guarantees should be expressed in terms of parameters that PB organizers have control over, such as the total budget and/or the range of possible project costs.

To illustrate these issues, it is instructive to  consider the simpler setting of approval-based multiwinner voting \citep{ABC+16a}, where all projects (referred to as candidates) have the same cost and $k$ of them need to be selected. In this setting, the first issue disappears entirely, as cost and cardinality satisfaction become equivalent. Moreover, the \emph{committee size} $k$ is a natural parameter. The best possible utilitarian guarantee of any proportional voting rule is approximately $\frac{2}{\sqrt{k}}$ \citep{LaSk20b, EFI+24a}, with this bound also being achieved by a variant of MES \citep{BrPe24a}. 

For PB, however, the picture is less clear. %
Prior work on utilitarian guarantees for proportional PB rules is extremely limited, with \citet{FVMG22a} providing the only study to date. However, their analysis yields only weak results: they derived bounds for various voting rules including rules satisfying EJR, but their guarantees are parameterized by the number of voters, making them vanishingly small for realistic PB instances. Moreover, their analysis was restricted to the cardinality satisfaction function (while most real-world instances use cost satisfaction) and produced lower and upper bounds on their utilitarian guarantees that are orders of magnitude apart. This led \citet{ReMa23a} to comment that ``the existing analysis of the price of fairness in PB \citep{FVMG22a} is still at a preliminary stage and there is a lot of room for improvement'' (page 85).

\subsection{Our Contribution}
We derive bounds on the utilitarian guarantee for \alg{MES}, and proportional rules in general, for a whole class of satisfaction functions.
To quantify our utilitarian guarantees, %
we use three parameters: 
the {budget} $b$, 
the {cost of the cheapest project} $c_{\min}$, and 
the {cost of the most expensive project} $c_{\max}$. 
In particular, we express our guarantees in terms of $\frac{b}{c_{\min}}$ and $\frac{b}{c_{\max}}$, which bound the number of projects that can be chosen in an exhaustive outcome, and are thus closely related to the committee size $k$ in multiwinner voting. To motivate our parametrization, we show that if project costs are allowed to vary arbitrarily, then we cannot get a meaningful utilitarian guarantee for \alg{MES} or even \alg{Greedy}, for any satisfaction function (\Cref{sec:negative_results}).
From a practical perspective, the budget and bounds on permissible project costs can be controlled by the PB organizer, while this is less true for other variables like the number of voters.

We derive positive results for a broad class of satisfaction functions known as \textit{DNS functions} (see \Cref{def:dns}), introduced by \citet{BFL+23a}. This class contains cost satisfaction and cardinality satisfaction as special cases. 
We also show that functions outside this class can lead to arbitrarily bad utilitarian ratios for \alg{MES}, no matter how much project costs are constrained (\Cref{prop:worstcase_non_dns}).

Our main result is a utilitarian guarantee for \alg{MES} (for any DNS satisfaction function). This guarantee is derived by comparing the welfare of the outputs of \alg{MES} and \alg{Greedy} (\Cref{sec:mes_greedy}), and combining that with a knapsack-inspired utilitarian guarantee for \alg{Greedy} to find a guarantee for \alg{MES} (\Cref{sec:mes_guarantee}). We show that this guarantee is asymptotically tight, not just for \alg{MES}, but for any proportional PB voting rule. Our bounds reduce to the multiwinner guarantees of \citet{BrPe24a} when all project costs are equal.

We also derive a utilitarian guarantee for the \alg{Greedy} rule for the curious case that the satisfaction function used to implement the rule differs from the satisfaction function used to measure voter welfare (\Cref{sec:greedy_greedy}).

\subsection{Related Work}

\paragraph{Utilitarian Guarantees in Multiwinner Voting.}
Besides the aforementioned work of \citet{FVMG22a}, our work is most closely related to the paper of \citet{BrPe24a}, who consider utilitarian guarantees in the multiwinner voting setting. They derive utilitarian guarantees for two proportional rules, showing that \alg{MES} achieves a guarantee of $\frac{2}{\sqrt{k}}-\frac{2}{k}$, and the Greedy Justified Candidate Rule (\alg{GJCR}) achieves a guarantee of $\frac{2}{\sqrt{k}}-\frac{1}{k}$. The latter matches the upper bound  on the utilitarian guarantee that is achievable by any proportional rule, as proven by \citet{LaSk20b} in the paper that initiated the study of utilitarian (and representation) guarantees in multiwinner voting. Moreover, \citet{EFI+24a} studied the impact of requiring the axiom of justified representation (a weakening of EJR) on these guarantees. %
Utilitarian guarantees were also studied by \citet{DBW+25a} and \citet{RML25a}, who analyzed how well one can approximate them together with other objectives.

\paragraph{Utilitarian Guarantees in Participatory Budgeting.}
Extending these multiwinner results to PB, however, proves challenging. \citet{FVMG22a} established bounds for various voting rules including those satisfying EJR. However, their analysis suffers from several limitations. First, their lower bound of $\Omega(\frac{1}{n \cdot c_{\max}})$ (with $c_{\min}$ normalized to~$1$) applies to all rules that exhaust the budget, regardless of proportionality. Second, their construction to prove the upper bound of $O(\frac{1}{\sqrt{n}})$ directly adapts an existing construction for multiwinner voting~\citep[Theorem~9]{LaSk20b} without exploiting the richer structure available in the PB setting. Notably, both bounds are parameterized by the number of voters $n$, making them practically meaningless for realistic PB instances: even with $n = 1{,}000$ voters,\footnote{At time of writing, 68\% of the instances in the PB data library \textit{Pabulib} %
\citep{FFP+23a} have $n\geq 1000$.} their lower bound guarantees less than $0.1\%$ of the optimal welfare. 

\paragraph{Participatory Budgeting.} %
PB has become a very active research area in computational social choice. We refer to \citet{ReMa23a} for an up-to-date survey on the topic. In particular, we are closely related to recent papers on proportionality in PB \citep[see, e.g.,][]{ALT18a, PPS21a, LCG22a, AzLe21a, Mal25a, ALS+24a}. Specifically, through our usage of the satisfaction function framework (and DNS functions in particular) to parameterize voting rules and welfare measures, we build on the work of \citet{BFL+23a}. Finally, we want to highlight the recent work of \citet{GPS+24a} who were motivated by the sometimes poor utilitarian performance of MES to design alternative voting rules, allowing for a better balance between proportionality and other desiderata.

\section{Preliminaries}

Let $P$ be a set of \emph{projects} and $N=\{1,\dots,n\}$ a set of \emph{voters}. For each voter $i \in N$, we let $A_i \subseteq P$ denote the set of projects approved by this voter. Together, $A =(A_i)_{i \in N}$ forms an \emph{approval profile}. For a project $p \in P$, we let $N_p = \{i\in N\colon p\in A_i\}$ denote the \emph{supporters of $p$}, i.e., the set of voters approving $p$. Finally, for each project $p \in P$ we are given a \emph{cost} $c(p) \in \mathbb{R}_{>0}$. Together with a budget limit $b \in \mathbb{R}_{>0}$, the tuple $I = (P,A,b,c)$ forms an approval-based participatory budgeting (PB) \textit{instance}. 
{We will use the PB instance depicted in \Cref{tab:running} as a running example in this section.}
Throughout the paper we will assume that $c(p) \le b$ for any project $p \in P$.
We let $\mathcal{I}$ be the set of all possible PB instances. A \emph{feasible outcome} for a given instance $I$ is any subset $P' \subseteq P$ with $c(P') = \sum_{p \in P'} c(p) \le b$. We call a feasible outcome $P'$ \emph{exhaustive} if there is no other feasible outcome $P'' \supset P'$.

\emph{Multiwinner} instances are a special class of PB instances in which projects all have the same cost, $c(p)=c$ for all $p\in P$, and the total budget is $b=kc$. Here, $k$ is also referred to as the \emph{committee size} (i.e., the number of projects to be~chosen).

\paragraph{Satisfaction Functions.}
Throughout the paper, we model voters' utilities using \emph{additive approval-based satisfaction functions} \citep{TaFa19a, BFL+23a}.
In order to argue about satisfaction and voting rules across different instances, we define satisfaction functions in an instance-agnostic way. %
In particular, we assume that there is a universe $\mathcal{P}$ of all possible projects, and we let $\mathcal{F}(\mathcal{P})$ be the set of all finite subsets of $\mathcal{P}$. A \emph{(global) satisfaction function} $\mu \colon \mathcal{F}(\mathcal{P}) \to \mathbb{R}_{\geq 0}$ maps subsets of projects to the satisfaction received from this subset. 
In line with \citet{BFL+23a}, we assume that satisfaction functions are \emph{additive}, i.e., it holds that $\mu(P') = \sum_{p \in P'} \mu(p)$ (with $\mu(p)$ being shorthand for $\mu(\{p\})$). We also assume that the satisfaction of any project $p$ is strictly positive, so that $\mu(P') =0$ if and only if $P'=\emptyset$.
For a given voter $i$, we let $\mu_i(P') = \mu(P' \cap A_i)$, i.e., the voter only receives satisfaction from the projects they approve.
The two most common satisfaction functions considered in the PB literature are the \emph{cost satisfaction function} $\mu^c$, with $\mu^c(p)=c(p)$ and $\mu^c(P')=c(P')$, and the \emph{cardinality satisfaction function} $\mu^\#$, with $\mu^\#(p)=1$ and $\mu^\#(P')=|P'|$. Note that these are already implicitly defined for all possible projects in an instance-agnostic manner.

Both $\mu^c$ and $\mu^\#$ belong to the larger class of \emph{weakly decreasing normalized satisfaction} functions \citep{BFL+23a}. 
\begin{definition}\label{def:dns}
    An additive satisfaction function $\mu$ is a \emph{weakly decreasing normalized satisfaction (DNS)} function if for any two projects $p,p' \in P$ with $c(p)\leq c(p')$, the following hold:
    \begin{enumerate}[label={(\arabic*)},leftmargin=0.65cm]
        \item $\mu(p)\leq \mu (p')$, i.e., the satisfaction of more expensive projects is weakly greater than that of cheaper projects; \label{dns1}
        \item $\frac{\mu(p)}{c(p)}\geq \frac{\mu(p')}{c(p')}$, i.e., the satisfaction per unit cost of more expensive projects is weakly smaller than that of cheaper projects. \label{dns2}
    \end{enumerate}
    
\end{definition}
The cost and cardinality satisfaction functions can be viewed as the two extreme examples of DNS functions: condition \ref{dns1} holds with equality for $\mu^\#$, while condition \ref{dns2} holds with equality for $\mu^c$. A different example of a DNS function would be $\mu(p) = \sqrt{c(p)}$.

Any DNS function is not just additive, but also \emph{cost-neutral} \citep{BFL+23a}: for any two projects $p, p'$ with $c(p)=c(p')$, it holds that $\mu(p)=\mu(p')$. %
Hence, the satisfaction of a project depends solely on the cost of the project. Note that in multiwinner voting instances, all cost-neutral satisfaction functions are equivalent to each other (up to normalization).

For a given satisfaction function $\mu$, we can define the \emph{utilitarian welfare} of a project $p$ as 
\(uw_\mu(p)=\sum_{i\in N}\mu_i(p)=|N_p| \cdot \mu(p).\)  
We further let $uw_\mu(P')=\sum_{p\in P'} uw_\mu(p)$ for any $P'\subseteq P$. 
For a project $p$ and satisfaction function $\mu$, we define the \emph{value} $v_\mu(p)$ as the ratio of $p$'s utilitarian welfare to its cost (``bang per buck''), i.e., $v_\mu(p)= \frac{uw_\mu(p)}{c(p)}=\frac{|N_p|\mu(p)}{c(p)}$.
When $\mu$ is clear from the context, we will omit the subscript. 

\setlength{\tabcolsep}{3pt}
\begin{table}[t]
\small

    \centering
    \begin{tabular}{cccccccccccc}
         Project $p$ & $c(p)$ & $A_1$ & $A_2$ & $A_3$ & $A_4$ & $A_5$ & $A_6$ & $A_7$ & $A_8$ & $A_9$ & $A_{10}$   \\
         \midrule
         $p_1$ & $65$ & \cmark & \cmark &\cmark & \cmark & \cmark & \cmark\\
         $p_2$ & $60$ & \cmark &  \cmark & \cmark & \cmark & \cmark \\
         $p_3$ & $40$ & \cmark & \cmark & \cmark &  & &  & \cmark\\
         $p_4$ & $20$ & & & & & & \cmark & \cmark & \cmark &  \\
         $p_5$ & $20$ & & & & & & & & & \cmark & \cmark \\

    \end{tabular}
    \caption{Costs and approval sets for \Cref{ex:running}.}
    \label{tab:running}
\end{table}

\begin{example}\label{ex:running}
Consider the PB instance with budget $b=100$ and $P=\{p_1,p_2,p_3,p_4,p_5\}$, with costs and approval sets given in \Cref{tab:running}. 
For the cost satisfaction function, we have $\mu^c(p_1)\!=\!c(p_1)\!=\!65$, $uw_{\mu^c}(p_1)\!=\!65\cdot 6\!=\!390$, and $v_{\mu^c} \!=\! 6$. 
For the cardinality satisfaction function, we have $\mu^\#(p_1)\!=\!1$, $uw_{\mu^\#}(p_1)\!=\!|N_{p_1}| \!=\! 6$, and $v_{\mu^\#} \!=\! 6/65\!\approx \!0.09$. %
\end{example}

\paragraph{Voting Rules.}
A \emph{voting rule} $\alg{R}$ is a function mapping each instance $I$ to a nonempty set of outcomes $\alg{R}(I)$.  Here, we do not just introduce individual voting rules, but families of rules, parameterized by the satisfaction function $\mu$ they use.
The first two voting rules we consider are \emph{sequential}: each begins with an empty bundle $P_0 = \emptyset$ and then iteratively adds projects to it, selecting project $p_k$ during \emph{stage} $k$. Let $P_k=\{p_1,\dots,p_k\}$ be the set of projects already chosen by some sequential rule \alg{R} after some stage $k$ in its execution. We say that at this point an unchosen project $p\in P\setminus P_k$ is \emph{affordable} if $c(P_k)+c(p)\leq b$.

As our first family of sequential voting rules, we consider \emph{greedy voting rules}. Let $\mu$ be a satisfaction function.  
For a given instance, the rule \algmu{Greedy} iteratively chooses the unchosen affordable project $p$ with the highest value $v_\mu(p)$ until no further project is affordable.

As our second family of rules, we study the \emph{Method of Equal Shares} \citep{PPS21a}.
For a given satisfaction function $\mu$, the rule \algmu{MES} works as follows. Each voter $i \in N$ is assigned a budget $b_i$, which is initialized to $\frac{b}{n}$. Then, at any given stage, for each 
$p \in P \setminus W$ compute $\rho(p)$ such that $\sum_{i\in N_p} \min(b_i,\rho(p)\cdot\mu_i(p))=c(p)$; if no such value exists, we set $\rho(p)=\infty$. If $\rho(p)=\infty$ for every project, \algmu{MES} stops. Otherwise, it adds a project $p \in \argmin_{p \in C\setminus W} \rho(p)$ to the outcome and updates the voter budgets to $b_i = b_i - \min(b_i,\rho(p) \mu_i(p))$ for all $i \in N_p$.

Finally, we also consider the \emph{maximum satisfaction rule}, which serves as the benchmark for our utilitarian guarantees. For a given satisfaction function $\mu$, the rule \algmu{MaxSat} selects all feasible outcomes $P'$ maximizing $\sum_{i \in N} \mu_i(P')$. In multiwinner instances, \alg{MaxSat} and \alg{Greedy} produce the same outcome.

One well-known shortcoming of \alg{MES} is that it is not exhaustive. To make it exhaustive, we consider \emph{completion rules}. A completion rule takes a feasible outcome $P'$ and supplements it with additional projects, outputting a feasible and exhaustive superset $P'' \supseteq P'$. We let $\overline{\algmu{MES}}$ be \algmu{MES} \emph{completed} by the \algmu{Greedy} rule. This refers to running \algmu{MES} on the full instance, and then, once \algmu{MES} terminates, running \algmu{Greedy} with the remaining projects and budget. 

For \Cref{ex:running}, the outcomes of our voting rules are as follows (employing cost satisfaction): \alg{Greedy}$_{\mu^c}$ selects $P_G\!=\!\{p_1,p_4\}$ with $uw(P_G)=450$; 
\alg{MaxSat}$_{\mu^c}$ selects $P^*\!=\!\{p_2,p_3\}$ with $uw(P^*)=460$; and
\alg{MES}$_{\mu^c}$ selects $P_M\!=\!\{p_3,p_4\}$ with $uw(P_M)=220$. 
$\overline{\alg{MES}_{\mu^c}}$ additionally selects $p_5$ and achieves $uw(\{p_3,p_4,p_5\})=260$.

\paragraph{Proportionality.} \alg{MES} is a rule that produces a proportional outcome, in the sense that it provides an appropriate level of representation to sufficiently cohesive groups. Following \citet{PPS21a}, we define the notion of cohesiveness below.

\begin{definition} \label{def:cohesiveness}
    Let $T\subseteq P$. A group $N'\subseteq N$ of voters is \emph{$T$-cohesive} if and only if $T\subseteq \bigcap_{i\in N'} A_i$ and $c(T)\leq \frac{|N'|}{n}b$.
\end{definition}

Using the notion of cohesiveness, we can define the following proportionality axiom. Just like our voting rules, the axiom is parameterized by a satisfaction function. %

\begin{definition}\label{def:EJR1} 
    An outcome $P'\subseteq P$ satisfies \emph{Extended Justified Representation up to one project with respect to $\mu$} (EJR1$_\mu$) if for every $T$-cohesive group $N'$, either $T\subseteq P'$ or there exists a voter $i\in N'$ and a project $p\in A_i \cap (P\setminus P')$ such that $\mu_i(P')+\mu_i(p)>\mu_i(T)$.
\end{definition}

We say that a rule \alg{R} satisfies EJR1$_\mu$ if the outcome of \alg{R} satisfies EJR1$_\mu$ for every instance $I\in \mathcal{I}$. %
\algmu{MES} satisfies EJR1$_\mu$ for any satisfaction function $\mu$ \citep{PPS21a}.\footnote{For $\mu=\mu^\#$, EJR1$_{\mu}$ is equivalent to the EJR axiom used by \citet{FVMG22a}.
If $\mu$ is a DNS function, the outcome of \algmu{MES} even satisfies the stronger \emph{EJR$_\mu$ up to any project} \citep{BFL+23a}. For the cost satisfaction function $\mu^c$, \alg{MES}$_{\mu^c}$ satisfies the even stronger \emph{EJR+$_{\mu^c}$ up to any project} \citep{BrPe23a}.} This is not the case for \algmu{Greedy}; for example, \alg{Greedy}$_{\mu^c}$ violates EJR1$_{\mu^c}$ in \Cref{ex:running} due to voter group $\{9,10\}$, which is $\{p_5\}$-cohesive. %

\paragraph{Utilitarian Guarantees.} Fix a satisfaction function $\mu$ and let \alg{R} be a voting rule. Let $P_R=\alg{R}(I)$ be the outcome of \alg{R} for a particular PB instance $I=(P,A,b,c)$, and let  $P^*=\algmu{MaxSat}(I)$. 
We say that \alg{R} 
has a \emph{utilitarian ratio} (w.r.t.~$\mu$) of $r\in [0,1]$ for instance $I$ if $\frac{uw_\mu(P_R)}{uw_\mu(P^*)}=r$.

In \Cref{ex:running}, the utilitarian ratio (w.r.t.~$\mu^c$) of \alg{Greedy}$_{\mu^c}$  is $\frac{450}{460}\approx 0.98$ and that of $\overline{\alg{MES}_{\mu^c}}$ is $\frac{260}{460}\approx 0.57$.

Now let $g\colon\mathcal{I}\rightarrow [0,1]$ be a function that can depend on properties of a particular PB instance $I\in \mathcal{I}$. We say that \alg{R} has a \emph{utilitarian guarantee} of $g$ if, for any $I\in \mathcal{I}$, the outcome $P_R=\alg{R}(I)$ has a utilitarian ratio of at least $g(I)$.

In principle, $g$ may depend on any properties of the instance. However, the more parameters it depends on, the less meaningful it will be. At one extreme, we could set the guarantee of rule \alg{R} to be the utilitarian ratio of \alg{R} for every instance. At the other extreme, if we do not allow $g$ to depend on any parameters then we cannot do better than $g=0$.\footnote{To see this, consider the multiwinner guarantee upper bound of $\frac{2}{\sqrt{k}}-\frac{1}{k}$ \citep{LaSk20b} for arbitrarily large $k$.} %

We will formulate our utilitarian guarantees using only the following properties of a PB instance $I=(P,A,b,c)$:

\begin{itemize}
    \item the instance budget $b(I)=b$; %
    \item the \emph{minimum project cost} $c_{\min}(I)=\min_{p\in P} c(p)$; and 
    \item the \emph{maximum project cost} $c_{\max}(I)=\max_{p\in P} c(p)$.
\end{itemize}

In fact, all of our results will actually depend on just two cost-based parameters, the \emph{normalized minimum} and \emph{maximum project costs}: $\frac{c_{\min}(I)}{b(I)}$ and $\frac{c_{\max}(I)}{b(I)}$. We can think of the reciprocals of these, $k_1(I)=\frac{b(I)}{c_{\max}(I)}$ and $k_2(I)=\frac{b(I)}{c_{\min}(I)}$, as providing lower and upper bounds (up to rounding) on the number of projects that can be contained in any exhaustive outcome. These generalize the multiwinner concept of committee size to the PB setting, with $k_1$ and $k_2$ corresponding to \emph{minimum} and \emph{maximum committee size}, respectively. %

Similarly, we define $\mu_{\min}(I)=\min_{p\in P} \mu(p)$ and $\mu_{\max}(I)=\max_{p\in P}\mu(p)$ as the minimum and maximum satisfaction of any project in a given instance $I$. In our notation, we often omit $I$ when the instance is clear from context.

\section{Negative Results}\label{sec:negative_results}

As a warm up, we demonstrate that certain natural assumptions are necessary to obtain meaningful utilitarian guarantees. In particular, the negative results in this section motivate both (i) our focus on costs as parameters for the guarantees and (ii) our restriction to DNS satisfaction functions. 

In this section, we let \algmu{PROP} be a rule satisfying EJR1$_\mu$. For instance, \algmu{PROP} could be \algmu{MES} with any completion rule. %
We begin by showing that, for any DNS satisfaction function $\mu$, we can construct instances such that \algmu{PROP}, and even \algmu{Greedy}, achieve an arbitrarily low utilitarian ratio. 

We first consider the case of satisfaction functions that are non-trivially bounded from below, like the cardinality satisfaction function $\mu^\#$.

\begin{proposition}\label{prop:worstcase_bounded}
    Let $\mu$ be a DNS function and $\varepsilon > 0$ a constant such that $\mu(p) > \varepsilon$ for all $p \in \mathcal{P}$.
    Then, \algmu{PROP} and \algmu{Greedy} have a utilitarian guarantee of at most $\frac{1}{n-1}$.
\end{proposition}
\begin{proof}
    Consider a PB instance with $n\geq 3$ voters and a project set $P$ containing two projects: project $p_1$ with $N_{p_1}=N\setminus \{1\}$ and $c(p_1)=b$, and
    project $p_2$ with $N_{p_2}=\{1\}$ and $c(p_2)= \frac{\varepsilon b}{n\mu(p_1)}$.
    This instance has two exhaustive outcomes, $\{p_1\}$ and $\{p_2\}$, as $c(p_1)+c(p_2)>b$.
Since $\mu$ is a DNS function, we know that $p_1$ has higher utilitarian welfare than $p_2$. Hence, \algmu{MaxSat} chooses the outcome $\{p_1\}$.
    Observe that 
    \(
    v(p_1) = \frac{(n-1)\mu(p_1)}{c(p_1)} = \frac{(n-1)\mu(p_1)}{b} < \frac{n\mu(p_1)}{b}
    \)
    and 
    \(
    v(p_2) = \frac{\mu(p_2)}{c(p_2)} > \frac{\varepsilon}{\frac{\varepsilon b}{n\mu(p_1)}} = \frac{n\mu(p_1)}{b}.
    \)
    Thus, $v(p_1)\! < \!v(p_2)$ and \algmu{Greedy} selects the outcome $\{p_2\}$. 
    As $\frac{\varepsilon}{\mu(p_1)} < 1$, we know that $c(p_2) < \frac{b}{n}$ and can thus be afforded by voter $1$. Thus, $\{1\}$ is a $\{p_2\}$-cohesive group, and as voter $1$ approves no other projects, we get that \algmu{PROP} must select $\{p_2\}$.
Using DNS assumption \ref{dns1}, this implies that the utilitarian ratio of \algmu{Greedy} and \algmu{PROP} is $\frac{uw(p_2)}{uw(p_1)}=\frac{\mu(p_2)}{(n-1)\mu(p_1)}\leq \frac{1}{n-1}$. 
\end{proof}

Therefore, we can choose an instance with a sufficiently large number of voters to make the utilitarian ratio of \alg{Greedy} and \alg{PROP} arbitrarily bad. Notably, \Cref{prop:worstcase_bounded} tightens the upper bound on the utilitarian guarantee (when parameterized only by $n$ only) for proportional rules with $\mu=\mu^\#$ from $O(1/\sqrt{n})$ \citep{FVMG22a} to $O(1/n)$.

Next, we consider the case of satisfaction functions that can get arbitrarily small, like the cost satisfaction function~$\mu^c$.

\begin{restatable}{proposition}{worstunb}\label{prop:worstcase_limiting}
    Let $\mu$ be a DNS satisfaction function such that for any $\varepsilon>0$ there exists a project $p \in \mathcal{P}$ with $\mu(p)<\varepsilon$. %
    Then, for any $\delta>0$, \algmu{Greedy} and \algmu{PROP} do not achieve a utilitarian guarantee of $\delta$. 
\end{restatable}

\Cref{prop:worstcase_bounded,prop:worstcase_limiting} combine to show that utilitarian guarantees for DNS satisfaction functions for both \alg{Greedy} and \alg{PROP} can be arbitrarily bad. It's worth emphasizing the fact that these results hold for \alg{Greedy} --- a rule whose explicit goal is to (greedily) maximize utilitarian welfare. This mirrors similar results for the greedy knapsack algorithm (see e.g. \citet{KPP04a}). 
The instances constructed in both proofs involve a very small project preventing the purchase of a larger one, which is reminiscent of the "Helenka Paradox" discussed by \citet{GPS+24a}. This suggests that we could potentially obtain better results for instances where project costs are closer together.

\smallskip 

We now motivate the idea of constraining ourselves to DNS satisfaction functions by showing that no matter how much we limit project costs (without making them all equal), we can find a non-DNS satisfaction function $\mu$ that would make any proportional rule achieve an arbitrarily low utilitarian ratio. %

\begin{restatable}{proposition}{wnondns}\label{prop:worstcase_non_dns}
Fix $k_1,k_2$ with $1\leq k_1<k_2$ and $\varepsilon > 0$. There exists a PB instance $I$ with $c_{\max}(I)=\frac{b}{k_1}$ and \mbox{$c_{\min}(I)=\frac{b}{k_2}$}, and a (non-DNS) satisfaction function $\mu$, such that the utilitarian ratio achieved by \algmu{PROP} is at most $\varepsilon$. 
\end{restatable}

Together, \Cref{prop:worstcase_bounded,prop:worstcase_limiting,prop:worstcase_non_dns} motivate the need for using costs as utilitarian guarantee parameters, and for constraining ourselves to DNS satisfaction functions. In the next sections, we will see that combining these allows us to obtain meaningful utilitarian guarantees for proportional rules. 

\section{Comparing the Welfare of MES and Greedy}\label{sec:mes_greedy}

In this section, we compare the utilitarian welfare generated by \algmu{MES} to that generated by \algmu{Greedy}, and bound the ratio of these from below. We can think of this as a \emph{comparative} utilitarian guarantee. We will use our results here to derive a utilitarian guarantee for \algmu{MES} in \Cref{sec:mes_guarantee}.

\begin{theorem}\label{thm:mesovergreedy}
Let $\mu$ be a DNS satisfaction function, and consider a PB instance $I$. Let $P_G$ and $P_M$ be the outcomes of \algmu{Greedy} and \GMmu{}, respectively, for this instance. Then,
$$\frac{uw(P_M)}{uw(P_G)}\geq 2\sqrt{\frac{c_{\min}(I)}{b(I)}}-\frac{c_{\min}(I)+c_{\max}(I)}{b(I)}.$$
\end{theorem}
This guarantee mirrors the $\frac{2}{\sqrt{k}} - \frac{2}{k}$ guarantee for multiwinner voting \citep{BrPe24a}.
Similarly to the proof of \citet{BrPe24a}, our proof utilizes the fact that \algmu{Greedy} and \algmu{MES} will select the same projects for a while, before diverging once voters start running out of their budgets during the execution of \algmu{MES}. In order to determine how early this divergence might occur, we introduce the concept of \emph{budget-limited} voters. %
\begin{definition}
    Let $P_k$ be the set of projects selected by \alg{MES} after stage $k$ of its execution. We say that voter $i$ is \emph{budget-limited} for an unchosen project $p\in A_i \setminus P_k$, if $p$ is affordable, but splitting its cost equally among its supporters would result in $i$ exceeding their personal budget. Formally, if $b_i$ is the remaining budget of voter $i$ after stage $k$, voter $i$ is budget-limited for $p$ if $p$ is affordable but $c(p)/|N_p| > b_i$. 
\end{definition}

We show that \algmu{MES} selects the same projects as \algmu{Greedy} until the stage in its execution where, for the first time, a voter is budget-limited for the affordable, unchosen project with highest value.

\begin{restatable}{lemma}{sameprojects}\label{lemma:sameprojects}

    Let $p_j$ be the project chosen in stage $j$ by \algmu{Greedy} and $p'_j$ be the project chosen in stage $j$ by \algmu{MES}. If for each $j\in \{1,\dots,k\}$, no voter is budget-limited for $p_j$ at stage $j$ of \algmu{MES}, then $p'_k=p_k$.
\end{restatable}

Using this, we can find the first \algmu{MES} stage during which a voter would be budget-limited for the project chosen by \alg{Greedy} at that stage. We call this stage $i+1$.\footnote{By defining this threshold stage in a more nuanced way than in the proof of Theorem 10 of \citet{BrPe24a}, we fix a minor error in their proof.} Using this stage as a threshold, we will consider three sets of projects and compare the utilitarian welfare generated by these:

\begin{enumerate}[label={(\arabic*)},leftmargin=0.65cm]
    \item the commonly chosen projects until stage $i$; \label{subsetone}
    \item the projects chosen by \algmu{Greedy} after stage $i$; and \label{subsettwo}
    \item the projects chosen by \GMmu{} after stage $i$. \label{subsetthree}
\end{enumerate}
Together, \ref{subsetone} and \ref{subsettwo} form $P_G$ and \ref{subsetone} and \ref{subsetthree} form $P_M$. 

In order to help bound the utilitarian welfare of set \ref{subsetthree}, we can find a lower bound for the value of any project selected by \algmu{MES}. Note that this does not restrict the value of projects selected by any completion rule we use to supplement the outcome of \algmu{MES}\,---\,we will consider those separately.

\begin{restatable}{lemma}{minvalue}\label{lemma:minvalue}
    Any project $p$ chosen by \algmu{MES} has a value of at least $v_\mu(p) \geq \frac{n}{b}\mu_{\min}$.
\end{restatable}

Moreover, we use the DNS assumption on $\mu$ to bound the ratio of satisfaction to cost in our instance.

\begin{restatable}{lemma}{dns}\label{lemma:dns}
    Let $\mu$ be a DNS satisfaction function and $I$ be a PB instance with project set $P$. Then, for any project $p\in P$, it holds that
    $\frac{\mu_{\min}(I)}{c_{\min}(I)}\geq \frac{\mu(p)}{c(p)} \geq \frac{\mu_{\max}(I)}{c_{\max}(I)}.$
\end{restatable}

We are now ready to prove \Cref{thm:mesovergreedy}.
\newcommand{\GMmunew}{$\overline{\mbox{\algmu{MES}}}$}
\begin{proof}[Proof of \Cref{thm:mesovergreedy}] 

We compare the projects chosen by \algmu{Greedy} to the projects chosen by \GMmu{}. 
Let 
\[P_G=\{p_1,\dots,p_g\} \quad  \text{and} \quad P_M=\{p'_1,\dots,p'_m\}\] 
be the sets of projects chosen by \algmu{Greedy} and \GMmu{}, respectively, indexed in the order they were chosen by those methods. If $P_G=P_M$, then $\frac{uw(P_M)}{uw(P_G)}=1$ and we are done, so we will assume that $P_G\neq P_M$. In particular, this 
means that $p'_1$ was chosen by \algmu{MES} rather than by its completion. 

Using \Cref{lemma:sameprojects}, we know that as long as no voters are budget-limited for the affordable, unchosen project with highest value, the \algmu{MES} part of \GMmu{} selects the same projects as \algmu{Greedy}. As $P_G\neq P_M$ we can let $p_{i+1}\in P_G$ be the first project chosen by \algmu{Greedy} where a voter would be budget limited for it after $\{p_1,\dots,p_i\}$ were selected by \algmu{MES}. Note that $i=0$ is possible, i.e., $p_{i+1}$ could be the very first project chosen. For notational convenience, let $v^*=v(p_{i+1})$ and $\mu^*=\mu(p_{i+1})$.
We observe that $p_j=p'_j$ for $1\leq j\leq i$ and let the total cost of those projects be $\val=\sum_{j=1}^i c(p_j)$. %

In order to compare the utilitarian welfare generated by \algmu{Greedy} and \GMmu{}, we will consider the welfare of the following subsets of projects: 
\eqref{projects1} The set $\{p_1,\dots,p_i\}$ of projects that are initially chosen by both \algmu{Greedy} and \GMmu{}; 
\eqref{projects2} the set $\{p_{i+1},\dots,p_g\}$ of projects chosen by \algmu{Greedy} after $p_i$; and 
\eqref{projects3} the set $\{p'_{i+1},\dots,p'_m\}$ of projects chosen by \GMmu{} after $p_i$.

We will prove the following bounds on the utilitarian welfare provided by these three projects sets:
\begin{align}
    &uw(\{p_1,\dots,p_i\}) \geq \val v^* \label{projects1} \\[0.4em]
    &uw(\{p_{i+1},\dots,p_g\})\leq(b-\val) v^* \label{projects2} \\
    &uw(\{p'_{i+1},\dots,p'_m\})\geq [b-\val-c_{\max}]\frac{n}{b}\mu_{\min} \label{projects3}
\end{align}

From our definition of value, we know that for any set of projects $P'\subseteq P$, we have $uw(P')=\sum_{p\in P'} c(p)v(p)$. 

For \eqref{projects1}, we observe that, by definition of \algmu{Greedy}, it holds that $v(p)\geq v^*$ for any $p\in \{p_1,\dots, p_i\}$ as any such project was picked before $p_i$. Hence, get $uw(\{p_1,\dots, p_i\}) = \sum_{j = 1}^i c(p_j) v(p_j) \ge \sum_{j = 1}^i c(p_j) v^* = \val v^*$.

For \eqref{projects2}, we observe that any project $p\in \{p_{i+1},\dots, p_g\}$ must have a weakly lower value than $p_{i+1}$. Since, moreover,  $c(\{p_{i+1},\dots, p_g\})$ is upper bounded by $b-\val$, we get that 
$uw(\{p_{i+1},\dots, p_n\}) = \sum_{j = i+1}^n c(p_j) v(p_j) \le \sum_{j = i+1}^n c(p_j) v^* \le (b - \val) v^*$.

For \eqref{projects3}, we construct a subset $P'_M\subseteq \{p'_{i+1},\dots,p'_m\}$ s.t. %
\addtocounter{equation}{-1}
\begin{subequations}
\begin{align}
    &c(P'_M)\geq b-\val-c_{\max} \label{subsetcost} \text{, and}\\
    &v(p)\geq \frac{n}{b}\mu_{\min} \text{ for each project $p\in P'_M$.} \label{subsetvalue}
\end{align}
\end{subequations}
Together, \eqref{subsetcost} and \eqref{subsetvalue} imply 
$uw(\{p'_{i+1},\dots,p'_m\})\geq uw(P'_M)\geq [b-\val-c_{\max}]\frac{n}{b}\mu_{\min}$.

The completion step of \GMmu{} considers projects that were not already chosen by \algmu{MES}, in order of decreasing value, and selects them as long as they are affordable. Let $p^*$ be the first unaffordable project that the completion step considers, and call the project chosen by \GMmu{} immediately before this point $p'_j$. Clearly $c(\{p'_1,\dots, p'_j\})+c(p^*)>b$, else $p^*$ would have been affordable. We define $P'_M=\{p'_{i+1},\dots,p'_j\}$, noting that $P'_M=\emptyset$ in the case  $j=i$.
Then, $c(P'_M)>b-\val-c(p^*)\geq b-\val-c_{\max}$, showing \eqref{subsetcost}. 

For \eqref{subsetvalue}, we consider $v(p)$ for any $p\in P'_M$. We have two cases: either $p$ was chosen by \algmu{MES} or the greedy completion. If $p$ was chosen by \algmu{MES} then we know from \Cref{lemma:minvalue} that $v(p)\geq\frac{n}{b} \mu_{\min}$. 

Let $P'\subseteq P$ be the set of projects with $v(p)\geq \frac{n}{b}\mu_{\min}$. 
Suppose that $p$ was chosen by the greedy completion, and, for a contradiction, assume that $v(p^*)<\frac{n}{b}\mu_{\min}$, which means that \GMmu{} selected every project $p$ with $v(p)\geq \frac{n}{b}\mu_{\min}$.  
Thus, $P'\subseteq \{p'_1,\dots,p'_j\}\subseteq P_M$. Furthermore, every project in $P_M\setminus P'$ must have been chosen by the greedy completion by \Cref{lemma:minvalue}.
As $c(P')\leq c(P_M)\leq b$, the outcome of \algmu{Greedy}, $P_G$, must also contain $P'$, as \algmu{Greedy} selects projects in decreasing order of value. Thus $P_G=P_M$, which contradicts our earlier assumption that $P_G\neq P_M$, and we can conclude that $v(p)\geq v(p^*)\geq \frac{n}{b}\mu_{\min}$, showing \eqref{subsetvalue}.

In order to compute the tradeoff between \algmu{Greedy} and \GMmu{} using \eqref{projects1}, \eqref{projects2}, and \eqref{projects3}, we determine $v^*$, i.e., the value of $p_{i+1}$. To do so, we first define the parameter
$\alpha = \frac{|N_{p_{i+1}}|b}{n c(p_{i+1})},$
representing how many times over the supporters of $p_{i+1}$ could buy $p_{i+1}$ with the budgets they are provided at the start of the execution of \algmu{MES}. %
Using the definition of $\alpha$ and $v^*$, we can rewrite the value of $p_{i+1}$ as 
$v^*=\frac{|N_{p_{i+1}}|{\mu^*}}{c(p_{i+1})}=\alpha {\mu^*}\frac{n}{b}.$

We distinguish three cases, based on the value of $\alpha$. In each of the three cases, we show that the 
bound $\frac{uw(P_M)}{uw(P_G)}\geq 2\sqrt{\frac{c_{\min}}{b}}-\frac{c_{\min} + c_{\max}}{b}$ holds. 

\noindent\textbf{Case 1:} $\alpha<\frac{\mu_{\min}}{{\mu^*}} \le 1$. We can find that every affordable project chosen by \GMmu{} after $p_i$ was chosen by the \algmu{Greedy} completion and thus $P_G=P_M$.\newline
\textbf{Case 2:} $\alpha\geq 1$. Using the fact that some voter was budget-limited for $p_{i+1}$ in step $i+1$ of \algmu{MES} we can bound the cost of $\{p_1,\dots,p_i\}$ with $\val\geq (\alpha-1)c_{\min}\frac{ {\mu^*}}{\mu_{\min}}$. We then combine this with our earlier results, and our assumption that $\mu$ is a DNS function, to derive the bound above.\newline
\textbf{Case 3:} $\frac{\mu_{\min}}{{\mu^*}}\leq \alpha< 1$. For this case we again combine our earlier results with the DNS assumptions.
\end{proof}

\section{The Utilitarian Guarantee of MES}\label{sec:mes_guarantee}

In this section, we combine a knapsack-inspired utilitarian guarantee for \alg{Greedy} with a modified version of \Cref{thm:mesovergreedy} to derive a utilitarian guarantee for \alg{MES}.

Using existing results from the knapsack literature (e.g., \citealp{KPP04a}),
we can derive the following utilitarian guarantee for the greedy rule. 

\begin{restatable}{proposition}{greedyguarantee}\label{prop:greedyguarantee}
Let $\mu$ be a (possibly non-DNS) satisfaction function. Then \algmu{Greedy} has a utilitarian guarantee of \mbox{$\frac{b-c_{\max}}{b}$}, and this guarantee is asymptotically tight.
\end{restatable}

One simple way to derive a utilitarian guarantee for \algmu{MES} would be to directly combine the results of \Cref{thm:mesovergreedy} and \Cref{prop:greedyguarantee}, obtaining a guarantee of $(2\sqrt{\frac{c_{\min}}{b}}-\frac{c_{min}+c_{\max}}{b})(\frac{b-c_{\max}}{b})$. However, in the PB setting, different rules do not necessarily spend the same proportion of the budget (or spend it efficiently). By simply multiplying our two bounds, we would be double-counting this inefficiency. By accounting for this more carefully, we can derive a utilitarian guarantee for \GMmu{} that coincides with the comparative guarantee from \Cref{thm:mesovergreedy}.

\begin{restatable}{theorem}{combinedguarantee}\label{thm:combined_guarantee}
    Let $\mu$ be a DNS satisfaction function. Then, \GMmu{} has a utilitarian guarantee of
$2\sqrt{\frac{c_{\min}}{b}}-\frac{c_{\min}+c_{\max}}{b}.$
\end{restatable}

\begin{proof}[Proof sketch]
    We first compare the output of \GMmu{} to a truncated output of \algmu{Greedy}, considering the subset of projects purchased by \algmu{Greedy} using the first $b-c_{\max}$ units of the budget that it spends. We lower bound the ratio of the utilitarian welfare of these two outcomes by $(2\sqrt{\frac{c_{\min}}{b}}-\frac{c_{\max}+c_{\min}}{b})\frac{b}{b-c_{\max}}$. The extra $\frac{b}{b-c_{\max}}$ factor reflects the decrease in total cost of the truncated greedy outcome compared to the original outcome of \algmu{Greedy}.
    We then compare the truncated \algmu{Greedy} outcome to the output of \algmu{MaxSat}, and find that this outcome has the same utilitarian ratio as guaranteed by \Cref{prop:greedyguarantee}: $\frac{b-c_{\max}}{b}$. Combining these bounds, we obtain a guarantee of $2\sqrt{\frac{c_{\min}}{b}}-\frac{c_{\min}+c_{\max}}{b}$.
\end{proof}

The result of \Cref{thm:combined_guarantee}, like all of the results in this paper, can be written in terms of just two instance parameters --- the minimum and maximum committee size of the instance: $2\sqrt{\frac{c_{\min}}{b}}-\frac{c_{\min}+c_{\max}}{b}=2\sqrt{\frac{1}{k_2}}-\frac{1}{k_2}-\frac{1}{k_1}$.

In order to give more context to our utilitarian guarantees, we consider the following simple example.

\begin{example}\label{ex:inpractice}

Let $\mu$ be any DNS function. Consider a PB instance class $\mathcal{I}^*\subseteq \mathcal{I}$ with $b(I)=\$1,000,000$, $c_{\min}(I)=\$10,000$ and $c_{\max}(I)= \$30,000$.
For this class, \algmu{MES} has a utilitarian guarantee of $2\sqrt{0.01}-0.01-0.03=0.16$. $\mathcal{I}^*$ is "similar" to a class of multiwinner elections with $k=\frac{b}{c_{\min}}=100$. (We could think of $\mathcal{I}^*$ as allowing some candidates to be bigger, taking up any\,---\,potentially non-integral\,---\,number of seats from $1$ to $3$.) For this election, the utilitarian guarantee of \alg{MES} is $0.18$ and that of \alg{GJCR} (another proportional rule) is $0.19$, using the bounds of \citet{BrPe24a}. 

\end{example}

Meanwhile, consider %
an \alg{MES} utilitarian guarantee that is a function of $n$, the number of voters. We know that for any DNS function $\mu$, this guarantee is bounded from above by $\frac{1}{n-1}$ (\Cref{prop:worstcase_bounded,prop:worstcase_limiting}). For instances with $n\!\ge\!1000$, this provides a utilitarian guarantee of at most~$0.001$.

\newcommand{\floorsqrt}[1]{\lfloor\sqrt{#1}\rfloor}

\smallskip

In the multiwinner setting, the \alg{MES} utilitarian guarantee is almost tight, with \citet{LaSk20b} deriving an upper bound of $\frac{2}{\floorsqrt{k}}-\frac{1}{k}$ for proportional rules. 
In the PB setting, we show that the \alg{MES} utilitarian guarantee from \Cref{thm:combined_guarantee} is asymptotically tight\,---\,not just for \alg{MES}, but for any rule satisfying EJR1\,---\,, at least for cost satisfaction.

\begin{restatable}{proposition}{tightness}\label{prop:tightness}

Let $\mu=\mu^c$ be the cost satisfaction function and let \algmu{PROP} be a rule satisfying EJR1$_\mu$. For each $k_1,k_2\in \mathbb{N}$ with $k_2\geq k_1$, there exists a PB instance $I$ with $c_{\min}(I)=\frac{b(I)}{k_2}$ and $c_{\max}(I)=\frac{b(I)}{k_1}$ such that \algmu{PROP} has a utilitarian ratio of at most %
$ \frac{2}{\floorsqrt{b/c_{\min}}}-\frac{c_{\min}+xc_{\max}}{b},$
where $x=\floorsqrt{k_2}\frac{c_{\min}}{c_{\max}}-\lfloor\floorsqrt{k_2}\frac{c_{\min}}{c_{\max}}\rfloor$.

\end{restatable}

When $c_{\min}=c_{\max}$, we get $x=0$ and the utilitarian ratio above reduces to the multiwinner bound. %

If we assume that $\sqrt{k_2}\in \mathbb{N}$, we get an upper bound of
$2\sqrt{\frac{c_{\min}}{b}}-\frac{c_{\min}+xc_{\max}}{b}.$ We note that $x< 1$ and that $x$ approaches $1$ for appropriate choices of $k_1$ and $k_2$. For instance, consider the sequences $k_1(a)=a-1$ and $k_2(a)=a^2$ for $a\in \mathbb{N}$, resulting in $x=\frac{a-1}{a} \rightarrow 1$.
For $x\rightarrow 1$, the utilitarian ratio in the proof above approaches the \alg{MES} utilitarian guarantee of $2\sqrt{\frac{c_{\min}}{b}}-\frac{c_{\min}+c_{\max}}{b}$ from \Cref{thm:combined_guarantee}. %

\section{Guarantees with Incorrect Satisfaction}\label{sec:greedy_greedy}

Finally, we consider the impact on welfare from using an ``incorrect'' satisfaction function for the \alg{Greedy} rule.

\begin{restatable}{theorem}{greedygreedy}\label{thm:greedy_greedy}
    Let $\mu$ be the actual DNS satisfaction function, and let $\mu'$ be another DNS satisfaction function. Then, the utilitarian guarantee  
    (as measured in terms of $\mu$)
    of \algmuprime{Greedy} is
    $\frac{b-c_{\max}}{b}\times\frac{c_{\min}}{c_{\max}}$. This bound is asymptotically tight. 
    
\end{restatable}
This guarantee is made up of two terms. The first, $\frac{b-c_{\max}}{b}$, captures the fact that \alg{Greedy} is only guaranteed to pick the most efficient set of projects for the first $b-c_{\max}$ dollars it spends. The second, $\frac{c_{\min}}{c_{\max}}$, can be thought of as the distortion of the instance, dictating how far DNS satisfaction functions can diverge. We can show that this bound is tight by constructing an example in which the actual satisfaction function is $\mu^c$ and our \alg{Greedy} rule is parametrised with $\mu^\#$.

\section{Conclusion and Future Work}
We studied the utilitarian guarantee of the Method of Equal Shares. In particular, we obtained a utilitarian guarantee of $2 \sqrt{\frac{c_{\min}}{b}} - \frac{c_{\max} + c_{\min}}{b}$ (with $c_{\min}$ and $c_{\max}$ being the minimum and maximum cost of a project and $b$ the budget limit in a given instance) when the satisfaction of voters is measured by a DNS satisfaction function. We further showed that this bound is tight for rules satisfying EJR1, and thus for MES.

There are multiple ways one could move forward from here. Firstly, in a very recent work, \citet{GPS+24a} introduced "MES with bounded overspending" in an attempt to fix some shortcomings of MES. Does this rule behave better than MES with regard to its utilitarian guarantee or are there perhaps other rules which allow to trade-off between proportionality and welfare? 

Secondly, our utilitarian guarantees %
are worst-case guarantees, and thus are unlikely to be representative of real-world scenarios. For instance, the experiments of \citet{BFNK19a}, \citet{EFI+24a}, and \citet{RML25a} indicate that proportional voting rules behave significantly better than what the worst-case guarantee would suggest. 
It would be interesting to understand to what extent empirical performance depends on the parameters of the instance, such as those used in our utilitarian guarantee.

\newpage 
\section*{Acknowledgments}

This work was supported by a \textit{Structural Democracy Fellowship} through the Brooks School of Public Policy at Cornell University, an \emph{EPSRC Doctoral Training Partnership} award, and by the Singapore Ministry of Education under grant number MOE-T2EP20221-0001.

\newpage
\appendix

\section*{Appendix}

\section{Omitted proofs}

\subsection*{Proofs from \Cref{sec:negative_results}}\label{app:negative_results}

\worstunb*
\begin{proof}
    Consider a PB instance with $n\geq 3$ voters and project set $P$ containing two projects:

    \begin{itemize}
    \item $p_1$, with $N_{p_1}=\{1\}$ and $c(p_1)=b$, and
    \item $p_2$, with $N_{p_2}=N\setminus \{1\}$, and $c(p_2) \le \frac{b}{n-1}$, chosen in such a way that $\mu(p_2) < \frac{1}{n-1} \mu(p_1)$. 
    \end{itemize}
    This instance has two exhaustive outcomes, $\{p_1\}$ and $\{p_2\}$. 

    As $uw(p_1) = \mu(p_1) > (n-1)\cdot\mu(p_2) = uw(p_2)$, \algmu{MaxSat} chooses the outcome $\{p_1\}$.

    Using DNS condition \ref{dns2} we get that 
    \begin{align*}
        v(p_2) = \frac{(n-1) \mu(p_2)}{c(p_2)} \ge \frac{(n-1) \mu(p_1)}{c(p_1)} > v(p_1)
    \end{align*} and hence, \algmu{Greedy} chooses project $p_2$. Similarly, as $c(p_2) \le \frac{b}{n-1}$ the project $p_2$ can be afforded by its supporters and hence  $N\setminus\{1\}$ is a $\{p_2\}$-cohesive group. As none of the voters in this group approve any other projects, we can conclude that \algmu{PROP} must select $\{p_2\}$.
    
    Thus, the utilitarian guarantee of \algmu{Greedy} and \algmu{PROP} is at most \[\frac{(n-1) \mu(p_2)}{\mu(p_1)}<\frac{\varepsilon\mu(p_1)}{\mu(p_1)}=\varepsilon. \qedhere\]
\end{proof}
\wnondns*
\begin{proof}

Consider a PB instance $I$ with $n\geq \frac{b}{c_{\max}-c_{\min}}$ voters and project set $P$ containing:

\begin{itemize}
    \item $p_j$ for $1\leq j\leq \lfloor\frac{b(n-1)}{n c_{\min}}\rfloor=k$, with $N_{p_j}=N\setminus \{1\}$ and $c(p_j)=c_{\min}$, and
    \item $p'$ with $|N_{p'}|=1$ and $c(p_j)=c_{\max}$.
\end{itemize}

We choose a cost-neutral satisfaction function $\mu$ with $\mu(p_j)=\varepsilon$ for $1\leq j\leq k$ and $\mu(p')=n^2$. This satisfaction function violates DNS condition \ref{dns2}. Let \algmu{PROP} be any rule that satisfies EJR1 for our chosen $\mu$.

\algmu{MaxSat} selects $P^*\supseteq \{p'\}$ with $uw(P^*)\geq n^2$.

Voter set $N\setminus \{1\}$ is a  $\{p_1,\dots,p_k\}$-cohesive group, and as these voters approve no other projects, the outcome of \alg{MES} and \alg{PROP} must contain $P_M=\{p_1,\dots,p_k\}$. Furthermore, using our $n\geq \frac{b}{c_{\max}-c_{\min}}$ assumption, we can find that $$c(P_M)=\lfloor\frac{b(n-1)}{n c_{\min}}\rfloor c_{\min}>\frac{b(n-1)}{n}-c_{\min}\geq b-c_{\max},$$ 
which means $P_M$ is an exhaustive outcome. $P_M$ is therefore the outcome of \alg{PROP}, with $uw(P_M)\leq n^2\varepsilon$. We can conclude that the utilitarian ratio of \algmu{PROP} for instance $I$ is at most $\varepsilon$, which we can choose to be arbitrarily small.
\end{proof}

\subsection*{Proofs from \Cref{sec:mes_greedy}}\label{app:mes_greedy}

\sameprojects*

\begin{proof}

    Whenever \alg{MES} chooses a project $p$ for which no voters are budget-limited, $p$'s cost is split equally among its set of supporters $N_p$. Thus we can find that it was chosen with the following $\rho$: $\rho(p)=\frac{c(p)}{|N_p|\mu(p)}=\frac{c(p)}{uw(p)}=\frac{1}{v(p)}$. Meanwhile, if a voter is budget-limited for $p$ then it cannot be purchased with equal payments and thus $\rho(p)>\frac{1}{v(p)}$. 
    
    Suppose that for each $j\in \{1,\dots,k\}$ no voter is budget-limited for $p_j$ at stage $j$ of \algmu{MES}. We show that if $p'_j=p_j$ for $j< m$, then $p'_m=p_m$ (this is the induction step). Using this assumption, we know that both \alg{Greedy} and \alg{MES} have selected $\{p_1,\dots, p_{m-1}\}$ before step $m$. Recall that \alg{Greedy} iteratively selects the affordable unchosen project with the highest value, and $\alg{MES}$ selects the project with the lowest $\rho$. Thus, we know that $v(p_m)\geq v(p)$ for any other unselected project $p\in P\setminus \{p_1,\dots, p_m\}$, and, since no voter is budget limited for $p_m$ in stage $m$ of \alg{MES}'s execution, we can conclude that $\rho(p_m)=\frac{1}{v(p_m)}\leq \rho(p)$ for any $p\in P\setminus \{p_1,\dots, p_m\}$, and therefore $p'_m=p_m$.

    We observe that the above argument also holds for $m=1$, and thus we can conclude that $p'_k=p_k$ by induction. 
\end{proof}

\minvalue*

\begin{proof}
    We know that $c(p)\leq |N_p|\frac{b}{n}$ as the cost of $p$ cannot exceed the total initial budgets of its supporters. Using this fact, $v(p)=|N_p|\frac{\mu(p)}{c(p)}\geq \frac{n}{b}\mu(p)\geq \frac{n}{b} \mu_{\min}$.
\end{proof}

\dns*

\begin{proof}    
    This follows directly from the definition of DNS functions.
    Using DNS condition \ref{dns1}, we know that $\mu_{\min}$ is the satisfaction of the cheapest project $p^c=\argmin_{p\in P}\{c(p)\}$ and $\mu_{\max}$ is the satisfaction of the most expensive project in $p^e=\argmin_{p\in P}\{c(p)\}$. Then, using DNS condition \ref{dns2}, we know that $\frac{\mu_{\min}}{c_{\min}}=\frac{\mu(p^c)}{c(p^c)}\geq \frac{\mu(p)}{c(p)} \geq \frac{\mu(p^e)}{c(p^e)}=\frac{\mu_{\max}}{c_{\max}}$.
\end{proof}

\begin{proof}[Proof of the Case Distinction for \Cref{thm:mesovergreedy}]
Recall the definition of $\alpha$,
\begin{equation}
\alpha = |N_{p_{i+1}}|\frac{b}{n c(p_{i+1})},    \label{thresholdalpha}
\end{equation} 
as well as
\begin{equation}
    v^*=\frac{|N_{p_{i+1}}|{\mu^*}}{c(p_{i+1})}=\alpha {\mu^*}\frac{n}{b}. \label{thresholdvalue}
\end{equation}

In each of the three cases we distinguished, we show that the bound $\frac{uw(P_M)}{uw(P_G)}\geq 2\sqrt{\frac{c_{\min}}{b}}-\frac{c_{\min} + c_{\max}}{b}$ holds. 

    \paragraph{Case 1:}  
    $\alpha<\frac{\mu_{\min}}{{\mu^*}} \le 1$. 
    
    In this case, for every affordable project $p\in P\setminus \{p_1,\dots,p_i\}$ it holds that $v(p)\leq v^*=\alpha {\mu^*}\frac{n}{b}< \frac{n}{b}\mu_{\min}$, using \eqref{thresholdvalue}. Thus by \Cref{lemma:minvalue}, every project in $P_M$ after $p_i$ could not have been picked by \alg{MES}, and must therefore have been picked by the greedy completion. As a consequence, we get that $P_M = P_G$ and thus also $\frac{uw(P_M)}{uw(P_G)} = 1$.
    
    \paragraph{Case 2:} 
    $\alpha\geq 1$.
    
    Consider a voter who was budget-limited for $p_{i+1}$ after step $i$ of \algmu{MES}. This voter must have had remaining funds less than $\frac{c(p_{i+1})}{|N_{p_{i+1}}|}=\frac{b}{\alpha n}$ (using \eqref{thresholdalpha}) and has thus spent at least $\frac{(\alpha-1)b}{\alpha n}$ on projects from among $\{p_1,\dots,p_i\}$. %
    We also know that each $p\in \{p_1,\dots,p_i\}$ has $v(p)\geq v^*$, and was funded equally by their supporters. Using these facts and the definition of $\alpha$ in \eqref{thresholdalpha} we find a lower bound for their total cost $\val$.
        \begin{align*}
        \val & \geq \frac{(\alpha-1)b}{\alpha n}\min_{p\in \{p_1,\dots,p_i\}}\{|N_p|\} \\
         & = \frac{(\alpha-1)b}{\alpha n}\min_{p\in \{p_1,\dots,p_i\}}\{v(p)\frac{c(p)}{\mu(p)}\} \\
         & \geq \frac{(\alpha-1)b}{\alpha n}v^*\min_{p\in \{p_1,\dots,p_i\}}\{\frac{c(p)}{\mu(p)}\} \\
         & = \frac{(\alpha-1)b}{\alpha n}|N_{p_{i+1}}|\frac{{\mu^*}}{c(p_{i+1})}\min_{p\in \{p_1,\dots,p_i\}}\{\frac{c(p)}{\mu(p)}\} \\
         & \geq (\alpha-1){\mu^*}\min_{p\in \{p_1,\dots,p_i\}}\{\frac{c(p)}{\mu(p)}\}
        \end{align*}
        
        From \Cref{lemma:dns} we know that $\min_{p\in \{p_1,\dots,p_i\}}\{\frac{c(p)}{\mu(p)}\}\geq \frac{c_{\min}}{\mu_{\min}}$ resulting in
    \begin{equation}
        \sum_{j=1}^i c(p_j)=\val\geq (\alpha-1)c_{\min}\frac{ {\mu^*}}{\mu_{\min}}. \label{totalcost}
    \end{equation}
    
    Let $\Ratio=\frac{uw(p'_1,\dots,p'_m)}{uw(p_1,\dots,p_g)}$ be the welfare ratio of \GM{} compared to \alg{Greedy}.
    We begin by using the inequalities in \eqref{projects1}, \eqref{projects2} and \eqref{projects3}, and simplifying:
    \begin{align*}
    \Ratio 
     & \geq \frac{\val v^*+(b-\val-c_{\max})\frac{n}{b}\mu_{\min}}{\val v^*+ (b-\val) v^*} \\
     & = \frac{\val (v^*-\frac{n}{b}\mu_{\min})+(b-c_{\max})\frac{n}{b}\mu_{\min}}{b v^*} \\
    \intertext{Using $v^*=\frac{\alpha {\mu^*}n}{b}$  \eqref{thresholdvalue} and $\val\geq \frac{(\alpha-1)c_{\min}{\mu^*}}{\mu_{\min}}$ \eqref{totalcost} :}
     \Ratio  
     &  \geq \frac{\frac{(\alpha-1)c_{\min} {\mu^*}}{\mu_{\min}} (\alpha {\mu^*}\frac{n}{b}-\frac{n}{b}\mu_{\min})+(b-c_{\max})\frac{n}{b}\mu_{\min}}{b \alpha {\mu^*}\frac{n}{b}} \\
    \intertext{Letting $M=\frac{\mu_{\min}}{{\mu^*}}$ (notably $ \frac{\mu_{\min}}{\mu_{\max}}\leq M\leq 1$), simplifying and rearranging to split into two terms:}
     \Ratio  
     &  \geq \frac{\alpha^2+\frac{b}{c_{\min}}M^2}{\alpha \frac{b}{c_{\min}} M}+\frac{\alpha(-1-M)+M-\frac{c_{\max}}{c_{\min}}M^2}{\alpha \frac{b}{c_{\min}} M}
    \end{align*}

    The first term can be simplified using the AM-GM inequality:
    
    $$\frac{\alpha^2+\frac{b}{c_{\min}}M^2}{\alpha \frac{b}{c_{\min}} M}\geq \frac{2\alpha \sqrt{\frac{b}{c_{\min}}}M}{\alpha \frac{b}{c_{\min}} M}=2\sqrt{\frac{c_{\min}}{b}}$$
    
    We can use \Cref{lemma:dns} to show that $c_{\min}\leq c_{\max}M$, and alongside with our assumption that $\alpha\geq 1$ we can find that 
        \begin{align*}
            & \frac{\alpha(-1-M)+M-\frac{c_{\max}}{c_{\min}}M^2}{\alpha \frac{b}{c_{\min}} M} \\
            &\geq  \frac{\alpha(-c_{\min}-c_{\min}M)+\alpha M[c_{\min}-c_{\max}M]}{\alpha b M} \\
            &=  \frac{-c_{\min}-c_{\max}M^2}{bM}\\
            &=  -\frac{1}{b}(\frac{c_{\min}}{M}+Mc_{\max})\\
            &=  -\frac{1}{b}(c_{\min}+c_{\max}+\frac{1-M}{M}c_{\min}-(1-M)c_{\max})\\
            &\geq -\frac{1}{b}(c_{\min}+c_{\max}+(1-M)c_{\max}-(1-M)c_{\max})\\
            &=  -\frac{c_{\min}+c_{\max}}{b}  
        \end{align*}
        The second last step above uses $c_{\min}\leq c_{\max}M$.

    Combining these results, we obtain
    $$\Ratio \geq 2\sqrt{\frac{c_{\min}}{b}}-\frac{c_{\min}+c_{\max}}{b}.$$
    
    \paragraph{Case 3:} 
    $\frac{\mu_{\min}}{{\mu^*}}\leq \alpha< 1$. 
    
    Using the inequalities in \eqref{projects1}, \eqref{projects2} and \eqref{projects3} and our characterization of $v^*$ in \eqref{thresholdvalue}, and simplifying:
    \begin{align*}
    \Ratio & =\frac{\val (v^*-\frac{n}{b}\mu_{\min})+(b-c_{\max})\frac{n}{b}\mu_{\min}}{b v^*} \\
     & =\frac{\val (\alpha {\mu^*}-\mu_{\min})+(b-c_{\max})\mu_{\min}}{\alpha b {\mu^*}} \\
     & \geq \frac{(b-c_{\max})\mu_{\min}}{b \mu_{\max}} \\
     & \geq \frac{\mu_{\min}}{\mu_{\max}}-\frac{c_{\max}}{b}\frac{\mu_{\min}}{\mu_{\max}}
    \end{align*}
    
    From DNS condition \ref{dns2} we know that $\frac{\mu_{\min}}{\mu_{\max}}\geq \frac{c_{\min}}{c_{\max}}$. Thus $\frac{\mu_{\min}}{\mu_{\max}}(\frac{b-c_{\max}}{b})\geq \frac{c_{\min}}{c_{\max}}-\frac{c_{\min}}{b}=(\frac{c_{\min}}{c_{\max}}+\frac{c_{\max}}{b})-\frac{c_{\min}+c_{\max}}{b}\geq 2\sqrt{\frac{c_{\min}}{b}}-\frac{c_{\min}+c_{\max}}{b}$ by the AM-GM inequality.
    
\end{proof}

\subsection*{Proofs from \Cref{sec:mes_guarantee}}\label{app:mes_guarantee}

\greedyguarantee*

Maximizing utilitarian welfare in the PB setting (for arbitrary additive utilities) is equivalent to solving the well-known Knapsack problem (see e.g. \citet{KPP04a} for an introduction). An instance of Knapsack contains a set of items each with a profit and weight, as well as a total knapsack capacity. A PB instance can be reduced to a knapsack instance by setting item profit to a project's utilitarian welfare, item weight to project cost, and total Knapsack capacity to the total budget. The \alg{Greedy} rule for PB is equivalent to the simplest Greedy algorithm for knapsack.

\begin{proof}

Let $P_G$ and $P^*$ be the outcomes of \algmu{Greedy} and \algmu{MaxSat} respectively, let $p^*$ is the first project that \algmu{Greedy} cannot afford and let $P'_G\subseteq P_G$ be the set of projects \algmu{Greedy} picked before considering $p^*$. 
The utilitarian guarantee follows from Knapsack literature (e.g. Corollary 2.2.2 in \citet{KPP04a}). 
We can find that $uw(P^*)-uw(P'_G)\leq [b-c(P'_G)]v(p^*)$ and $b-c(P'_G)\leq c(p^*)\leq c_{\max}$, and for each $p\in P'_G, v(p)\geq v(p^*)$ which yields: %

\begin{align*}
\frac{uw(P_G)}{uw(P^*)} & \geq \frac{uw(P'_G)}{uw(P^*)} \\
 & \geq \frac{uw(P'_G)}{uw(P'_G)+v(p^*)(b-c(P'_G))} \\
 & =1-\frac{v(p^*)(b-c(P'_G))}{uw(P'_G)+v(p^*)(b-c(P'_G))} \\
 & \geq 1-\frac{v(p^*)c_{\max}}{v(p^*)c(P'_G)+v(p^*)(b-c(P'_G))} \\
 & =\frac{b-c_{\max}}{b}
\end{align*}

We can show that this guarantee are tight using the following example. Let $\mu$ be the cost satisfaction function and let $\varepsilon>0$ and $x\in \naturals^+$ arbitrary. Consider a PB instance with $n\geq 2$ and project set $P$ containing:
\begin{itemize}
    \item $p_j$ for $1\leq j\leq x-1$, with $c(p_j)=(1+\varepsilon)\frac{b}{x}$, and $N_{p_j}=N$.
    \item $p'_j$ for $1\leq j\leq x$, with $c(p'_j)=\frac{b}{x}$, and $N_{p'_j}=N\setminus \{1\}$
\end{itemize}

Observe that $c_{\max}=(1+\varepsilon)\frac{b}{x}$, $c_{\min}=\frac{b}{x}$, and thus project costs are arbitrarily close together, without being exactly equal.
\algmu{Greedy} selects $P_G=\{p_1,\dots, p_{x-1}\}$, with $uw(P_G)=\frac{x-1}{x}bn(1+\varepsilon)$. 
\algmu{MaxSat} selects $P^*=\{p'_1,\dots, p'_x\}$, with $uw(P^*)=b(n-1)$. 

$$\frac{uw(P_G)}{uw(P^*)}=\frac{n}{n-1}\frac{x-1}{x}(1+\varepsilon)=\frac{n}{n-1}(1+\varepsilon-\frac{c_{\max}}{b})$$

Thus, in the $n\rightarrow \infty$ and $\varepsilon\rightarrow 0$ limit, $\frac{uw(P_G)}{uw(P^*)}$ converges to the (minimum) utilitarian guarantee in \Cref{prop:greedyguarantee}.
\end{proof}

\combinedguarantee*

We prove \Cref{thm:combined_guarantee} by comparing the output of \alg{MES} to a truncated output of \alg{Greedy}. To reason about such a truncated output, we define the following concept.

\begin{definition}
    Let $p\in P$ be any project in some instance $I$. We let $\overline{p}$ be the \emph{fractional part} of $p$ on which $c^*<c(p)$ dollars was spent. We define $c(\overline{p})=c^*$, $N_{\overline{p}}=N_p$ and $\mu(\overline{p})=\frac{c(\overline{p})}{c(p)}$, resulting in $v(\overline{p})=v(p)$.
\end{definition}

Note that $\mu(\overline{p})$ is defined in order to cause no change to the value of $p$, and need not follow our DNS (or cost-neutrality) assumption.\footnote{$\overline{p}$ isn't a real project, and thus $\overline{p}\notin \mathcal{P}$.} We use the idea of a fractional project to derive the following comparative guarantee.

\begin{restatable}{lemma}{mesovergreedytruncated}\label{lemma:mesovergreedytruncated}
    Let $\mu$ be a DNS satisfaction function and assume $I$ is a PB instance with $c_{\max}<b$. Consider the projects purchased by \algmu{Greedy} using the first $b-c_{\max}$ dollars it spends, purchasing a fraction of the project that the $(b-c_{\max})$th dollar is spent on. Let this set of projects be $\hat{P_G}=\{p_1,\dots, p_{\hat{g}-1},\overline{p_{\hat{g}}}\}$, with $c(\overline{p_{\hat{g}}})=b-c_{\max}-c(\{p_1,\dots, p_{\hat{g}-1}\})$
    
    Let $P_M$ be the outcome of \GMmu{}, as in \Cref{thm:mesovergreedy}. Then:

$$\frac{uw(P_M)}{uw(\hat{P_G})}
\geq\bigg(2\sqrt{\frac{c_{\min}}{b}}-\frac{c_{\max}+c_{\min}}{b}\bigg)\frac{b}{b-c_{\max}}$$

\end{restatable}

\begin{proof}
    In our proof of \Cref{thm:mesovergreedy} we identified three subsets of projects, and derived welfare bounds for them. These are restated below.
    \setcounter{equation}{0}
    \begin{align}
    &uw(\{p_1,\dots,p_i\}) \geq \val v^* \label{projects1_} \\
    &uw(\{p_{i+1},\dots,p_g\})\leq(b-\val) v^* \label{projects2_} \\
    &uw(\{p'_{i+1},\dots,p'_m\})\geq [b-\val-c_{\max}]\frac{n}{b}\mu_{\min} \label{projects3_}
    \end{align}

    If $\val>b-c_{\max}$ then $P_M\supseteq \{p_1,\dots, p_{\hat{g}}\}$, and thus $\frac{uw(P_M)}{uw(\hat{P_G})}\geq 1$. We can use this to find that $2\sqrt{\frac{c_{\min}}{b-{c_{\max}}}}-\frac{c_{\max}+c_{\min}}{b-c_{\max}}\leq 2\sqrt{\frac{c_{\max}}{b-{c_{\max}}}}-\frac{c_{\max}}{b-c_{\max}}\leq 1\leq \frac{uw(P_M)}{uw(\hat{P_G})}$.
    Thus, we henceforth assume $\val\leq b-c_{\max}$. In order to account for truncating $P_G$ to $\hat{P_G}$ we modify \eqref{projects2_} to consider the projects from $\hat{P_G}$ chosen after $p_i$ and obtain the following:
    \setcounter{equation}{1}
    \begin{align}
        uw(\{p_1,\dots, p_{\hat{g}-1},\overline{p_{\hat{g}}}\})\leq(b-\val-c_{\max}) v^* \refstepcounter{equation}\tag{\theequation*} \label{projects2_modified}
    \end{align}
    
    Let's consider how this affects our case distinction, letting $\Ratiohat=\frac{uw(p'_1,\dots,p'_m)}{uw(\{p_1,\dots, p_{\hat{g}-1},\overline{p_{\hat{g}}}\})}$:

    \paragraph{Case 1:} $P_G=P_M$ and thus $\Ratiohat\geq 1$.

    \paragraph{Case 2:} The truncation affects the denominator by replacing $b$ with $b-c_{\max}$:

    \begin{align*}
        \Ratiohat 
        & \geq \frac{\val (v^*-\frac{n}{b}\mu_{\min})+(b-c_{\max})\frac{n}{b}\mu_{\min}}{(b-c_{\max}) v^*} \\
        & \geq \frac{b}{b-c_{\max}}\times \frac{\val (v^*-\frac{n}{b}\mu_{\min})+(b-c_{\max})\frac{n}{b}\mu_{\min}}{b v^*}
    \end{align*}

    The rest of the proof follows from the same steps as the proof of \Cref{thm:mesovergreedy}

    \paragraph{Case 3:}

    \begin{align*}
        & \frac{uw(P_M)}{uw(\{p_1,\dots, p_{\hat{g}-1},\overline{p_{\hat{g}}}\})}\\ 
        =\ & \frac{\val (v^*-\frac{n}{b}\mu_{\min})+(b-c_{\max})\frac{n}{b}\mu_{\min}}{(b-c_{\max}) v^*} \\
         \geq\ & \frac{\mu_{\min}}{\mu_{\max}} 
    \end{align*}

    We know from the proof of \Cref{thm:mesovergreedy} that $\frac{\mu_{\min}}{\mu_{\max}}\frac{b-c_{\max}}{b}\geq  2\sqrt{\frac{c_{\min}}{b}}-\frac{c_{\min}+c_{\max}}{b}$, and thus $\frac{\mu_{\min}}{\mu_{\max}}\geq (2\sqrt{\frac{c_{\min}}{b}}-\frac{c_{\min}+c_{\max}}{b})\frac{b}{b-c_{\max}}$ holds also.
\end{proof}

We use this result to derive a more precise utilitarian guarantee for MES.

\begin{proof}[Proof of \Cref{thm:combined_guarantee}]
    Let $P^*$ be the set of projects selected by \alg{MaxSat}. We know from \Cref{prop:greedyguarantee} that $\frac{uw(P_G)}{uw(P^*)}\geq\frac{b-c_{\max}}{b}$. We can use a similar argument to show that $\frac{uw(\hat{P_G})}{uw(P^*)}\geq\frac{b-c_{\max}}{b}$, by observing that the proof of \Cref{prop:greedyguarantee} works for the truncated output $\hat{P_G}$ also.\footnote{Intuitively we know that \algmu{Greedy} spends the first $b-c_{\max}$ dollars optimally}

    Putting this together with the result from \Cref{lemma:mesovergreedytruncated} we obtain the guarantee above.
\end{proof}

\tightness*

\begin{proof}

    Let $I$ be a PB instance with $n=k_2=\frac{b}{c_{\min}}$, and $P$ consisting of the following projects.
    \begin{itemize}
    \item $p_j$ for $1\leq j\leq k_1$, with $c(p_j)=c_{\max}$ and $N_{p_j}=\{1,\dots,{\floorsqrt{n}\}}\}$.
    \item $p'_j$ for $\floorsqrt{n}+1\leq j\leq n$, with $c(p'_j)=c_{\min}$, with $N_{p'_j}=\{j\}$.
    \end{itemize}

    \alg{MaxSat} selects $P^*=\{p_1,\dots,p_{k_1}\}$ with $c(P^*)=b$ and $uw(P^*)=\floorsqrt{n} b=\floorsqrt{n}nc_{\min}$. 
    
    For \alg{PROP} we notice that for each $\floorsqrt{n}+1\leq j\leq n$, $\{j\}$ is a $\{p'_j\}$-cohesive group, and voter $j$ approves no other projects, so \alg{PROP} must select each $p'_j$. \alg{PROP} also selects $\lfloor\floorsqrt{n}\frac{k_1}{n}\rfloor$ projects from among $\{p_1,\dots,p_{k_1}\}$. Call this set of projects $P_M$. This outcome is exhaustive as 
    \begin{align*}
        c(P_M)
        &=(n-\floorsqrt{n})\frac{n}{b}+\lfloor\floorsqrt{n}\frac{k_1}{n}\rfloor c_{\max}\\
        &= (n-\floorsqrt{n})\frac{n}{b}+(\floorsqrt{n}\frac{k_1}{n}-x)c_{\max}\\
        &=b-xc_{\max}>b-c_{\max},
    \end{align*}

    which means no more projects can be selected. We can find that $uw(P_M)= (n-\floorsqrt{n})c_{\min}+\floorsqrt{n}(\lfloor\floorsqrt{n}\frac{k_1}{n}\rfloor) c_{\max}$ 
    
    We can now compute the utilitarian ratio of \alg{PROP}.
    \begin{align*}
    \frac{uw(P_M)}{uw(P^*)}
    &\leq \frac{1}{\floorsqrt{n}}-\frac{1}{n}
    +(\floorsqrt{n}\frac{k_1}{n}-x)\frac{c_{\max}}{nc_{\min}} \\
    &\leq \frac{1}{\floorsqrt{n}}-\frac{1}{n}
    +\frac{\floorsqrt{n}\frac{k_1}{n}\frac{b}{k_1}}{nc_{\min}}
    -x\frac{c_{\max}}{nc_{\min}} \\
    &\leq \frac{2}{\floorsqrt{n}}
    -\frac{c_{\min}+xc_{\max}}{nc_{\min}}\\
    &=\frac{2}{\floorsqrt{b/c_{\min}}}-\frac{c_{\min}+xc_{\max}}{b}
    \end{align*}
\end{proof}

\subsection*{Proofs from \Cref{sec:greedy_greedy}}\label{app:greedy_greedy}

\greedygreedy*

\begin{proof}
    Let $\mu$ be a (real) DNS satisfaction function, with $uw=uw_\mu$ and $v=v_\mu$. Let $\mu'$ be another DNS satisfaction function with corresponding value function $v_{\mu'}$. We compare the outcomes of \algmu{MaxSat} and \algmuprime{Greedy}, calling them $P^*$ and $P_G$ respectively. If $c(P)\leq b$, then $P^*=P_G=P$, so we suppose $c(P)> b$. 
    
    We know that $c(P^*)\leq b$ and $c(P_G)> b-c_{\max}$, as \algmuprime{Greedy} is exhaustive and $c(P)>b$. 
    There exists $P'_G\subseteq P_G$ with $c(P'_G)>b-c_{\max}$ and $v_{\mu'}(p_g)\geq v_{\mu'}(p^*)$ for any $p_g\in P'_G$ and any $p^*\in P^*\setminus P_G$, else \algmuprime{Greedy} would have selected some project from $P^*\setminus P_G$.

    Let $p_g\in P'_G$ and $p^*\in P^*\setminus P_G$.
    \begin{align*}
    \frac{v(p_g)}{v(p^*)}
    &=\frac{\frac{|N_{p_g}|}{c(p_g)}\mu(p_g)}{\frac{|N_{p^*}|}{c(p^*)}\mu(p^*)}\\
    &=\frac{v_{\mu'}(p_g)\mu(p_g)\mu'(p^*)}{v_{\mu'}(p^*)\mu(p^*)\mu'(p_g)}\\
    &\geq \frac{\mu(p_g)}{\mu(p^*)}\frac{\mu'(p^*)}{\mu'(p_g)}
    \end{align*}

    We know from the definition of DNS functions that for any $p_1,p_2\in P$ with $c(p_1)\leq c(p_2)$ we have $\frac{\mu(p_2)}{\mu(p_1)}\geq 1$ and $\frac{\mu(p_1)}{\mu(p_2)}\geq \frac{c(p_1)}{c(p_2)}\geq \frac{c_{\min}}{c_{\max}}$. This means that regardless of whether $c(p_g)\leq c(p^*)$, or $c(p_g)>c(p^*)$ we can find that 
    $\frac{v(p_g)}{v(p^*)}
    \geq\frac{\mu(p_g)}{\mu(p^*)}\frac{\mu'(p^*)}{\mu'(p_g)}
    \geq \frac{c_{\min}}{c_{\max}}$.
    
    Let $v^{max}=\max_{p^*\in P^*\setminus P_G} v(p^*)$. Clearly $v(p_g)\geq \frac{c_{\min}}{c_{\max}}v_{\max}$ for any $p_g\in P'_G$. Using this we can bound the utility of the following project subsets.
    \setcounter{equation}{0}
    \begin{align*}
    &uw(P^*\cap P'_G)\geq \frac{c_{\min}}{c_{\max}}v^{max}c(P^*\cap P'_G) %
    \\
    &uw(P'_G\setminus P^*)\geq \frac{c_{\min}}{c_{\max}}v^{max}c(P'_G\setminus P^*) %
    \\
    &uw(P^*\setminus P'_G)\leq v^{max}c(P^*\setminus P'_G)%
    \end{align*}

    Finally, we use these bounds to find the utilitarian guarantee of \algmuprime{Greedy}.    
    \begin{align*}
        \frac{uw(P_G)}{uw(P^*)}
        & \geq\frac{uw(P'_G)}{uw(P^*)}\\
        & \geq \frac{\frac{c_{\min}}{c_{\max}}c(P'_G)}{c(P^*\setminus P'_G)+\frac{c_{\min}}{c_{\max}}c(P'_G\cap P^*)} \\
        & \geq \frac{c_{\min}}{c_{\max}}\frac{c(P'_G)}{c(P^*)} \\
        & \geq \frac{c_{\min}}{c_{\max}}\frac{b-c_{\max}}{b}
    \end{align*}

    To show tightness, we consider the following example, using $\mu=\mu^c$ and $\mu'=\mu^\#$.

    \smallskip
    
    Fix $b,k_1,k_2\in \mathbb{N}^+$, with $k_1<k_2$, and let $n$ be the number of voters in the instance. Let $c_{\max}=\frac{b}{k_1}$ and $c_{\min}=\frac{b-c_{\max}+\varepsilon}{k_2}$, with $\varepsilon>0$ sufficiently small. Note that this is a slight departure from our usual definition of $k_2$.

    Define project set $P$ to include the following projects:
    \begin{itemize}
        \item For $1\leq j\leq k_1$, project $p_j$  with $c(p_j)=c_{\max}$ and $|N_{p_j}|=n$.
        \item For $1\leq j\leq k_2$, project $p_j'$ with $c(p_j')=c_{\min}$ and $|N_{p_j'}|=\lfloor n\frac{c_{\min}}{c_{\max}}\rfloor +1=\lceil n\frac{c_{\min}}{c_{\max}} \rceil$ (we can guarantee $n\frac{c_{\min}}{c_{\max}}$ is not an integer by choosing $\varepsilon\in \mathbb{Q}$). 
    \end{itemize}

    \noindent \alg{MaxSat}$_{\mu^c}$ selects $P^*=\{p_1,\dots,p_{k_1}\}$ with $uw(P^*)=bn$.
    On the other hand, 
    \alg{Greedy}$_{\mu^\#}$ selects $P_G=\{p_1',\dots,p_{k^2}'\}$ %
    with $uw(P_G)=(b-c_{\max}+\varepsilon)(\lceil n\frac{c_{\min}}{c_{\max}}\rceil)$. We can choose an appropriately large $n=n_\varepsilon$ (given our already fixed $\varepsilon$, $k_1$ and $k_2$) to make $\lceil n_\varepsilon\frac{c_{\min}}{c_{\max}}\rceil\leq (n_\varepsilon +\varepsilon)\frac{c_{\min}}{c_{\max}}$.

    Thus, $\frac{uw(P_G)}{uw(P^*)}\leq \frac{b-c_{\max}+\varepsilon}{b}\frac{c_{\min}}{c_{\max}}\frac{n_\varepsilon+\varepsilon}{n_\varepsilon}$, which converges to $\frac{b-c_{\max}}{b}\frac{c_{\min}}{c_{\max}}$ as required.
\end{proof}

\end{document}